\documentclass[letterpaper, 10 pt, conference]{ieeeconf}  

\IEEEoverridecommandlockouts                              

\usepackage{cite}
\usepackage{amsmath,amssymb,amsfonts}
\usepackage{algorithm}
\usepackage{algpseudocode}
\usepackage{graphicx}
\usepackage{textcomp}
\usepackage{subcaption}

\title{\LARGE \bf
Physics--Guided Neural Networks for Feedforward Control: From Consistent Identification to Feedforward Controller Design*
}

\author{Max Bolderman$^{1}$, Mircea Lazar$^{1}$, and Hans Butler$^{1,2}$
\thanks{\begin{flushright}
\begin{minipage}[r]{0.04\textwidth}
    \includegraphics[width=0.9\linewidth]{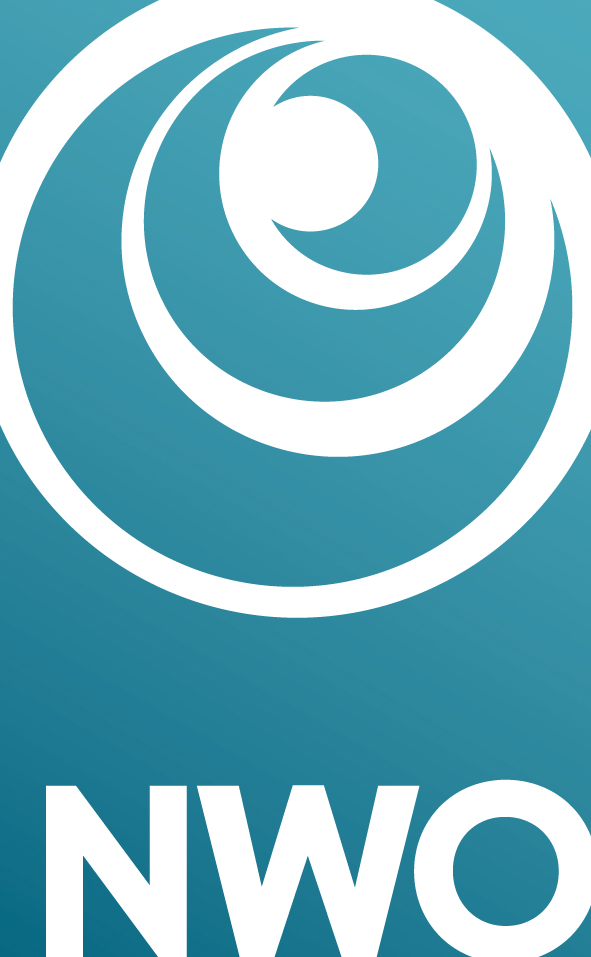}
\end{minipage}
\end{flushright}
\vspace{-1.05cm}
\begin{minipage}[l]{0.43\textwidth}
\hspace{0.5em} *This work is part of the research programme 9654 with project number 17973, which is (partly) financed by the Dutch Research Council (NWO). 
\end{minipage}}
\thanks{$^{1}$Control Systems Group, department of Electrical Engineering, Eindhoven University of Technology, The Netherlands. {\tt\small \{m.bolderman, m.lazar\}@tue.nl}}%
\thanks{$^{2}$ASML, Veldhoven, The Netherlands. 
        {\tt\small hans.butler@asml.nl}}%
}

\begin{document}
\newtheorem{thm}{Definition}[section]
\newtheorem{proposition}[thm]{Proposition}
\newtheorem{lemma}[thm]{Lemma}
\newtheorem{remark}{Remark}[section]
\newtheorem{definition}{Definition}
\newtheorem{assumption}{Assumption}[section]

\maketitle
\thispagestyle{empty}
\pagestyle{empty}

\begin{abstract}
	Model--based feedforward control improves tracking performance of motion systems, provided that the model describing the inverse dynamics is of sufficient accuracy. 
	Model sets, such as neural networks (NNs) and physics--guided neural networks (PGNNs) are typically used as flexible parametrizations that enable accurate identification of the inverse system dynamics. 
	Currently, these (PG)NNs are used to identify the inverse dynamics directly. 
	However, direct identification of the inverse dynamics is sensitive to noise that is present in the training data, and thereby results in biased parameter estimates which limit the achievable tracking performance. 
	In order to push performance further, it is therefore crucial to account for noise when performing the identification. 
	To address this problem, this paper proposes the use of a forward system identification using (PG)NNs from noisy data.
	Afterwards, two methods are proposed for inverting PGNNs to design a feedforward controller for high--precision motion control. 
	The developed methodology is validated on a real--life industrial linear motor, where it showed significant improvements in tracking performance with respect to the direct inverse identification.
\end{abstract}

\section{INTRODUCTION}
\label{sec:Introduction}
Model--based feedforward control strategies significantly improve tracking performance of motion systems, provided that the available model describing the inverse system dynamics is of sufficient accuracy~\cite{Devasia2002}. 
Typically, physics--based models, i.e., model sets that are obtained based on physical knowledge of the considered system, are used for feedforward control, see, e.g.,~\cite{Boerlage2003, Jamaluding2009, Igarashi2021}. 
However, the limited approximation capabilities of physics--based models result in structural model errors when such models are used to describe the complete dynamical behaviour~\cite{Schoukens2019, Nguyen2018}. 
Certainly, this becomes apparent when considering manufacturing tolerances, parasitic friction forces, or electromagnetic disturbances that are omnipresent in motion systems~\cite{Nguyen2018, Ma2018}. 

To deal with parasitic forces or other nonlinear phenomena that are hard to model, neural networks (NNs) are proposed as parametrizations for identification of the inverse system dynamics in~\cite{Sorensen1999}, see also~\cite{Zhang2007, Wanigasekara2019}. 
The universal approximation capabilities of NNs theoretically enables a perfect description of the system dynamics if the NN dimensions are chosen sufficiently large, and sufficiently exciting data can be obtained~\cite{Hornik1991}.
Additionally, it was shown in~\cite{Bolderman2021, Bolderman2022} that augmenting a physics--based feedforward with NNs improves tracking performances and robustness to non--training data. 

However, inverse model--based feedforward control design, including NN or PGNN model parametrizations, hinges on correctly identifying the inverse dynamics, which requires noiseless data.
Indeed, if noise is present in the training data, inverse identifications are prone to result in biased parameter estimates~\cite{Jung2013}. 
These biases decrease the model accuracy, and thereby introduce limitations on the achievable tracking performance resulting from the feedforward~\cite{Devasia2002}. 
On the other hand, methods for identification in the presence of noise are mostly designed for identification of the original, forward system dynamics~\cite{Nelles2001, Ljung1999}. 

Therefore, in this paper we address the problem of nonlinear model--based feedforward design from noisy data using a forward dynamics identification approach. 
First, we discuss the required assumptions to obtain a consistent estimate of the forward dynamics using PGNN model parametrizations for relevant noise structures. 
Then, we propose two methods for inverting the identified forward dynamics parametrized using (PG)NNs: a gradient--based numerical method that is inspired by techniques that are discussed in~\cite{Jensen1999}, and an analytic method suitable for a class of electromagnetic actuators common in high--precision motion control. 
Note that it is also possible to obtain a model of the inverse by performing a second identification based on noise--free data that is generated from the identified forward model~\cite{Widrow2008}. 
However, this requires two separate identifications, where each identification involves an experiment design, a parametrization, and a non--convex optimization problem. 

An extra benefit is that identification of the forward dynamics, i.e., minimization of the difference between the measured output and the predicted output, is in line with the evaluation of the tracking performance achieved by the feedforward controller. See also~\cite{Aarnoudse2021} which achieves this by filtering the cost function with the process sensitivity in a direct inverse identification setting, i.e., minimization of the difference between the input and the predicted input. 

The main contributions of this paper are as follows.
First, based on fundamental approaches discussed in~\cite{Ljung1999} for general nonlinear system identification, we show that forward system identification using a PGNN model class results in parameter estimates that are consistent, i.e., unbiased with probability $1$ when the data length goes to infinity. 
It is shown that the estimates remain consistent when the data is generated from a closed--loop experiment. 
Secondly, we derive methods for inversion of the identified PGNN describing the forward dynamics. 
Initially, a gradient--based technique is suggested to find the feedforward input.
Afterwards, an analytically invertible PGNN model is proposed for the case when the gradient--based technique is not implementable in real--time.

The remainder of this paper is organized as follows: Section~\ref{sec:Preliminaries} introduces the preliminaries, followed by the problem statement in Section~\ref{sec:ProblemStatement}. 
Section~\ref{sec:ConsistentEstimation} discusses the forward system identification using PGNNs. Section~\ref{sec:FeedforwardControllerDesign} derives the PGNN feedforward controller starting from the identified forward PGNN model. An experimental validation is performed in Section~\ref{sec:ExperimentalValidation}, followed by the main conclusions in Section~\ref{sec:Conclusions}.

\section{PRELIMINARIES}
\label{sec:Preliminaries}

\subsection{System dynamics and feedforward control}
Consider the discrete--time, single--input single--output (SISO), nonlinear time--invariant system with autoregressive exogeneous (ARX) noise structure, such that
\begin{align}
\begin{split}
\label{eq:SystemDynamics}
	y (t) & = h \big( \phi(t) \big) + v(t), \\
	\phi(t) & = [y(t-1), \hdots, y(t-n_a), \\
	& \quad \quad \quad u(t-n_k-1), \hdots, u(t-n_k-n_b)]^T. 
\end{split}
\end{align}
In~\eqref{eq:SystemDynamics}, $y(t)$ is the system output at time index $t$, $\phi(t)$ the regressor, $u(t)$ the input, $n_a$ and $n_b$ describe the order of the dynamics, and $n_k$ is the number of pure input delays. 
The function $h: \mathbb{R}^{n_a+n_b} \rightarrow \mathbb{R}$ describes the system dynamics, and $v(t)$ is assumed to be a zero mean white noise with variance $\sigma_v^2 := \lim_{N \to \infty} \frac{1}{N} \sum_{t = 0}^{N-1} v(t)^2$. 

\begin{remark}
	The nonlinear ARX (NARX) structure as in~\eqref{eq:SystemDynamics} is mostly popular for its simplicity, and does not generally describe the noise experienced on practical applications, see Fig.~\ref{fig:NoiseStructures}. 
	Although the main derivations in this paper focus on the NARX setting, we also highlight the output--error (NOE) and input--error (NIE) due to their direct relation with sensor and actuator noise, respectively.  
	Other noise structures can also be considered in a similar way. 
\end{remark}

From Fig.~\ref{fig:NoiseStructures}, we observe that the NOE dynamics gives
\begin{align}
\begin{split}
\label{eq:SystemDynamicsOE}
	y(t) & = h \big( \phi_{\textup{NOE}}(t) \big) + v(t), \\
	\phi_{\textup{NOE}}(t) & = [y(t-1)-v(t-1), \hdots, y(t-n_a)-v(t-n_a), \\
	& \quad \quad \quad u(t-n_k-1), \hdots, u(t-n_k-n_b)]^T.
\end{split}
\end{align}
Similarly, the NIE dynamics is given as
\begin{align}
\begin{split}
\label{eq:SystemDynamicsIE}
	y(t) & = h \big( \phi_{\textup{NIE}}(t) \big), \\
	\phi_{\textup{NIE}}(t) & = [y(t-1), \hdots, y(t-n_a), u(t-n_k-1)+\\
	 v(t-&n_k-1), \hdots, u(t-n_k-n_b) + v(t-n_k-n_b)]^T.
\end{split}	
\end{align}

The feedforward input $u_{\textup{ff}}(t)$ is the input $u(t)$ that yields $y(t) = r(t)$ for system~\eqref{eq:SystemDynamics} and some desired reference signal $r(t)$, when $v(t) = 0$. 
Substitution of $y(t) = r(t)$ and $u(t) = u_{\textup{ff}}(t)$ in~\eqref{eq:SystemDynamics} and shifting both sides $n_k+1$ samples forward in time, gives
\begin{align}
\begin{split}
\label{eq:OptimalFeedforward}
	r(t+n_k+1) = h \big( [& r(t+n_k), \hdots, r(t+n_k-n_b+1), \\
	& u_{\textup{ff}}(t), \hdots, u_{\textup{ff}}(t-n_b+1)]^T \big).
\end{split}
\end{align}
Then, with a slight abuse of notation, let $h^{-1}$ be the mapping describing the inverse dynamics such that the optimal feedforward is
\begin{align}
\begin{split}
\label{eq:OptimalFeedforwardRewritten}
	u_{\textup{ff}} (t) =&  h^{-1} \big( [ r(t+n_k+1), \hdots, r(t+n_k-n_b+1), \\
	& \quad \quad \; \; u_{\textup{ff}}(t-1), \hdots, u_{\textup{ff}}(t-n_b+1)]^T \big).
\end{split}
\end{align}
However, the actual function $h$ is unknown and, therefore, it cannot be used to design the feedforward controller. 

In the remainder of this paper, we use $\phi(t)$ as the regressor for the forward dynamics, and $\phi'(t)$ as the regressor for the inverse dynamics, e.g., $\phi'(t) = [y(t+n_k+1), \hdots, y(t+n_k-n_b+1), u(t-1), \hdots, u(t-n_b+1)]^T$ for the NARX case. 
Similarly, $\phi_{\textup{ff}}(t)$ and $\phi_{\textup{ff}}'(t)$ are obtained by subsitution of $y(t) = r(t)$ and $u(t) = u_{\textup{ff}}(t)$ in $\phi(t)$ and $\phi'(t)$, respectively.

\begin{figure}
	\centering
	\includegraphics[width=1.0\linewidth]{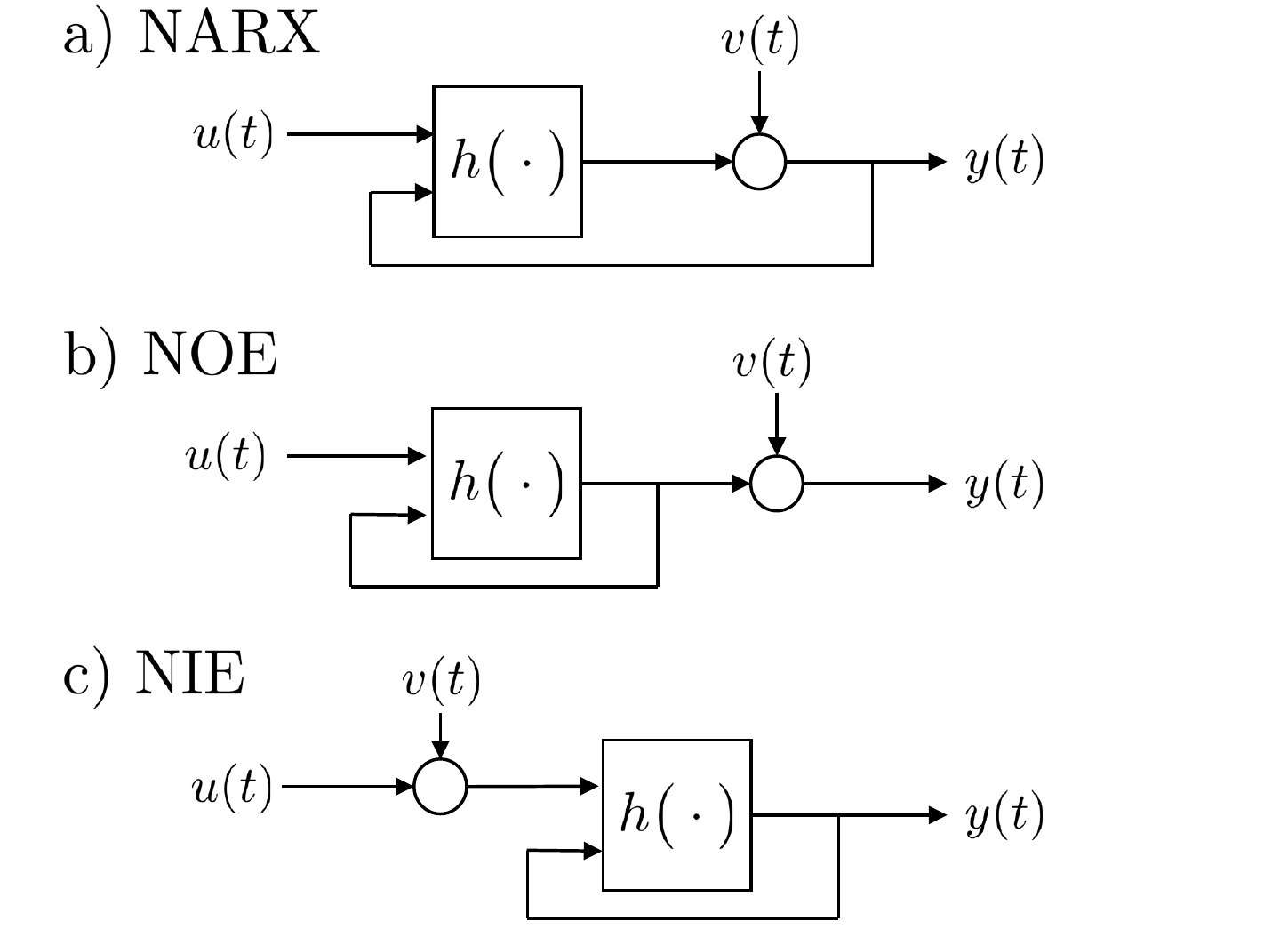}
	\caption{Schematic illustration of the different noise structures typically encountered for high--precision mechatronics. }
	\label{fig:NoiseStructures}
\end{figure}

\subsection{System identification procedure}
Typically, a physics--based model parametrization is derived from first principle modelling.
\begin{definition}
	A physics--based model is defined as
	\begin{equation}
	\label{eq:PhysicsBasedModel}
		\hat{y} \big( \theta_{\textup{phy}}, \phi(t) \big) = f_{\textup{phy}} \big( \theta_{\textup{phy}}, \phi(t) \big),
	\end{equation}
	where $\hat{y} \big( \theta_{\textup{phy}}, \phi(t) \big)$ indicates the prediction of the output $y(t)$, and $\theta_{\textup{phy}}$ are the parameters of the physical model. 
\end{definition}

The parameters $\theta_{\textup{phy}}$ are chosen according to an identification criterion, such as the mean--squared error (MSE). 
\begin{definition}
	The MSE identification criterion is given as
	\begin{equation}
	\label{eq:IdentificationCriterionMSE}
		\theta_{\textup{phy}}^* = \textup{arg} \min_{\theta_{\textup{phy}}} \frac{1}{N} \sum_{t = 0}^{N-1} \left( y(t) - \hat{y} \big( \theta_{\textup{phy}}, \phi(t) \big) \right)^2,
	\end{equation}
	where the summation is taken over a data set $Z^N = \{ \phi(0), y(0), \hdots, \phi(N-1), y(N-1) \}$. 
\end{definition}

Following the same reasoning as for~\eqref{eq:OptimalFeedforwardRewritten}, after identification of the parameters $\theta_{\textup{phy}}^*$, the physics--based feedforward controller is given as
\begin{equation}
\label{eq:FeedforwardPhysicsBased}
	u_{\textup{ff}}(t) = f_{\textup{phy}}^{-1} \big( \phi_{\textup{ff}}' (t) \big), 
\end{equation}
where $f_{\textup{phy}}^{-1}$ indicates the inverse of $f_{\textup{phy}}$, which is assumed known (typically a physics--based analytic formula is used, e.g., inverse linear motion dynamics). 
In order to obtain an implementable physics--based feedforward as in~\eqref{eq:FeedforwardPhysicsBased}, we assume knowledge of the reference up until time $t+n_k+1$ at time $t$, and assume that $f_{\textup{phy}}^{-1}$ is bounded--input bounded--output (BIBO) stable. 
Relevant methods for obtaining a stable inverse of linear systems that have unstable inverses are listed in~\cite{Zundert2017}. 

In general, the physics--based model~\eqref{eq:PhysicsBasedModel} does not capture the actual dynamics $h$ in~\eqref{eq:SystemDynamics}, due to the presence of parasitic friction, electromagnetic distortions, and other effects present in the system~\cite{Nguyen2018}. 
\begin{definition}
	The unmodelled dynamics are defined as $g \big( \phi(t) \big) := h \big( \phi(t) \big) - f_{\textup{phy}} \big( \theta_{\textup{phy}}^*, \phi(t) \big)$, such that the system~\eqref{eq:SystemDynamics} is rewritten into
	\begin{equation}
	\label{eq:SystemDynamicsUnmodelledDynamics}
		y(t) = f_{\textup{phy}} \big( \theta_{\textup{phy}}^*, \phi(t) \big) + g \big( \phi(t) \big) + v(t). 
	\end{equation}
\end{definition}

As suggested in~\cite{Bolderman2021}, in order to identify also the unmodelled dynamics, we augment the physics--based model with a NN to obtain the PGNN. 
\begin{definition}
	A PGNN is defined as
	\begin{equation}
	\label{eq:PGNN}
		\hat{y} \big( \theta, \phi(t) \big) = f_{\textup{phy}} \big( \theta_{\textup{phy}}, \phi(t) \big) + f_{\textup{NN}} \big( \theta_{\textup{NN}}, \phi(t) \big), 
	\end{equation}
	with $\theta := \{ \theta_{\textup{phy}}, \theta_{\textup{NN}} \}$ the PGNN parameters, and $\theta_{\textup{NN}} := \{ W_1, B_1, \hdots, W_{l+1}, B_{l+1} \}$ the NN weights and biases with $l$ the number of hidden layers. The NN output is
	\begin{equation}
	\label{eq:NN}
		f_{\textup{NN}} \big( \theta_{\textup{NN}}, \phi(t) \big) = W_{l+1} \alpha_l \big( \hdots \alpha_1 \big( W_1 \phi(t) + B_1 \big) \big) + B_{l+1}, 
	\end{equation}
	where $\alpha_i$ denotes the aggregation of activation functions of layer $i= 1, \hdots , l$. 
\end{definition}
\begin{remark}
	The PGNN becomes a standard, black--box NN when no physical knowledge of the system is present, i.e., for $f_{\textup{phy}} = 0$ in~\eqref{eq:PGNN}. 
	Therefore, the methodology developed in this paper applies also to black--box NNs, with the exception of the analytic inverse in Section~\ref{sec:AnalyticInverse}. 
\end{remark}

The flexible nature of the NN can create an overparameterization in the PGNN~\eqref{eq:PGNN} when training according to the MSE identification criterion~\eqref{eq:IdentificationCriterionMSE}, which results in a parameter drift during training. 
Therefore, a regularized MSE identification criterion was introduced in~\cite{Bolderman2022}. 
\begin{definition}
	The regularized MSE identification criterion is given as
	\begin{align}
	\begin{split}
	\label{eq:IdentificationCriterionMSERegularized}
		\hat{\theta} = \textup{arg} \min_{\theta} & \frac{1}{N} \sum_{t = 0}^{N-1} \left( y(t) - \hat{y} \big( \theta, \phi(t) \big) \right)^2 \\
		& + \left( \theta_{\textup{phy}}^* - \theta_{\textup{phy}} \right)^T \Lambda \left( \theta_{\textup{phy}}^* - \theta_{\textup{phy}} \right), 
	\end{split}
	\end{align}
	with $\Lambda$ a positive definite matrix, and $\theta_{\textup{phy}}^*$ the solution of~\eqref{eq:IdentificationCriterionMSE} for the physical model contained within the PGNN~\eqref{eq:PGNN}. 
\end{definition}

The majority of literature on (PG)NN--based feedforward control performs a direct inverse identification, see, e.g.,~\cite{Bolderman2021, Zhang2007}. 
Basically, the inverse dynamics is parametrized
\begin{equation}
\label{eq:PGNNInverse}
	\hat{u} \big( \theta, \phi'(t) \big) = f_{\textup{phy}}^{-1} \big( \theta_{\textup{phy}}, \phi'(t) \big) + f_{\textup{NN}} \big( \theta_{\textup{NN}}, \phi'(t) \big),
\end{equation}
where $\hat{u} \big( \theta, \phi(t) \big)$ is the predicted input. 
Then, the parameters $\theta$ are trained according to identification criterion
\begin{align}
\begin{split}
\label{eq:InverseIdentificationCriterionRegularized}
	\hat{\theta} = \textup{arg} \min_{\theta} & \frac{1}{N} \sum_{t = 0}^{N-1} \left( u(t) - \hat{u} \big( \theta, \phi'(t) \big) \right)^2 \\
		& + \left( \theta_{\textup{phy}}^* - \theta_{\textup{phy}} \right)^T \Lambda \left( \theta_{\textup{phy}}^* - \theta_{\textup{phy}} \right).
\end{split}
\end{align} 
In general, if the data is noise free, this approach is more attractive for feedforward control design, as we directly obtain the inverse dynamics. 
However, when the data contains noise, it is observed that parameter estimates are biased, as illustrated in the next section.

\subsection{Inverse identification and induced parameter bias}
For simplicity of exposition, consider a linear ARX system, i.e., $ h \big( \phi(t) \big) = \theta_0^T \phi(t)$ in~\eqref{eq:SystemDynamics}. Then, the inverse system dynamics is given as
\begin{equation}
\label{eq:LinearInverseDynamics}
	u(t) = {\theta_0'}^T \phi'(t) - \frac{1}{\psi_0} v(t+n_k+1),
\end{equation}
with $\psi_0$ the parameter that multiplies $u(t-n_k-1)$ in the forward dynamics, and $\theta_0'$ the parameters obtained after inversion. 
We parametrize~\eqref{eq:LinearInverseDynamics} as
\begin{equation}
\label{eq:LinearInverseDynamicsParameterization}
	\hat{u} \big( \theta', \phi'(t) \big) = {\theta'}^T \phi'(t),
\end{equation}
and identify the parameters according to MSE criterion~\eqref{eq:InverseIdentificationCriterionRegularized} with $\Lambda =0$.
We define $M:=\frac{1}{N} \sum_{t=0}^{N-1} \phi'(t) \phi'(t)^T$, such that the least--squares solution is given as
\begin{equation}
\label{eq:LinearLeastSquaresSolution}
	\hat{\theta}' = M^{-1} \frac{1}{N} \sum_{t = 0}^{N-1} \phi'(t) u(t), 
\end{equation}
where $M$ must be non--singular to return a unique estimate, i.e., the data set $Z^N$ must be persistently exciting. 
Substitution of~\eqref{eq:LinearInverseDynamics} in~\eqref{eq:LinearLeastSquaresSolution} and computing the bias for $N \to \infty$ gives
\begin{align}
\begin{split}
\label{eq:LinearLeastSquaresBias}
	\lim_{N \to \infty} \hat{\theta}'-\theta_0' = &\lim_{N \to \infty} M^{-1}   \frac{1}{N} \sum_{t = 0}^{N-1} \phi'(t) \frac{1}{\psi_0} v(t+n_k+1) .
\end{split}
\end{align}
The bias is nonzero, since $M$, and therefore also $M^{-1}$, is non--singular and the second term is nonzero due to the correlation between $y(t+n_k+1)$ in $\phi'(t)$ and $v(t+n_k+1)$.


\section{PROBLEM STATEMENT}
\label{sec:ProblemStatement}
Since direct inverse system identification results in a biased estimate in the presence of noisy data, even when the system follows a basic linear ARX structure, model--based feedforward controller design from noisy data remains an open problem. 

Our aim is to develop a systematic PGNN feedforward controller design procedure based on the following two steps:
\begin{enumerate}
	\item \emph{Consistent parameter estimation:} the PGNN consistently identifies the forward dynamics, including the unmodelled dynamics $g \big( \phi(t) \big)$ in the presence of noise;
	\item \emph{System inversion:} in order to derive the PGNN feedforward, the identified forward dynamics must be either analytically or numerically invertible. 
\end{enumerate}


\section{CONSISTENT PGNN IDENTIFICATION}
\label{sec:ConsistentEstimation}
Since the universal approximation theorem for NNs holds only within a compact domain~\cite{Hornik1991}, we define the operating conditions of the feedforward controller.
\begin{definition}
	The operating conditions $\mathcal{R}$ are defined as all possible regressors provided to the PGNN, such that
	\begin{equation}
	\label{eq:OperatingConditions}
		\phi_{\textup{ff}} (t), \phi(t) \in \mathcal{R},
	\end{equation}
	for all $t$, all references supplied to the PGNN feedforward, and all regressors in the data set $Z^N$. 
\end{definition}

Then, following the fundamental framework in~\cite{Ljung1999}, it is possible to obtain consistent estimates of the system~\eqref{eq:SystemDynamics}. 
\begin{definition}
	A parameter estimate $\hat{\theta}$ of $\theta^*$ is consistent if $\hat{\theta} \to \theta^*$ for $N \to \infty$ with probability $1$. 
\end{definition}

We adopt the following common assumptions on the model, data, and training to prove consistency for the PGNN identification. 
\begin{assumption}
\label{as:ModelSet}
	There exists a $\theta_{\textup{NN}}^*$ such that $f_{\textup{NN}} \big( \theta_{\textup{NN}}, \phi(t) \big) = g \big( \phi(t) \big)$ for all $\phi(t) \in \mathcal{R}$. 	
\end{assumption}
\begin{assumption}
\label{as:PersistenceOfExcitation}
	For $\theta_{\textup{NN}}^A \neq \theta_{\textup{NN}}^B$ with $f_{\textup{NN}} \big( \theta_{\textup{NN}}^A, \phi(t) \big) \neq f_{\textup{NN}} \big( \theta_{\textup{NN}}^B, \phi(t) \big)$ for some $\phi(t) \in \mathcal{R}$, we have
	\begin{equation}
	\label{eq:PersistenceOfExcitation}
		\frac{1}{N} \sum_{t = 0}^{N-1} \left( f_{\textup{NN}} \big( \theta_{\textup{NN}}^A, \phi(t) \big) - f_{\textup{NN}} \big( \theta_{\textup{NN}}^B, \phi(t) \big) \right)^2 > 0. 
	\end{equation}
\end{assumption}
\begin{assumption}
\label{as:GlobalOptimum}
	The optimization over $\theta$ of the identification criterion~\eqref{eq:IdentificationCriterionMSERegularized} yields a global optimum.
\end{assumption}

\begin{proposition}
\label{prop:PGNNConsistencyARX}
	Consider the PGNN~\eqref{eq:PGNN} that is used to identify the NARX system~\eqref{eq:SystemDynamics} according to identification criterion~\eqref{eq:IdentificationCriterionMSERegularized}. 
	Suppose that Assumptions~\ref{as:ModelSet}, \ref{as:PersistenceOfExcitation} and~\ref{as:GlobalOptimum} hold.
	Then, for $N \to \infty$, the identified PGNN parameters satisfy $\hat{\theta} = \{ \hat{\theta}_{\textup{phy}}, \hat{\theta}_{\textup{NN}} \} \to \{ \theta_{\textup{phy}}^*, \theta_{\textup{NN}}^* \}$. 
\end{proposition}
\begin{proof}
	The proof follows the approach in~\cite{Ljung1976}, by showing that the globally minimizing argument of the cost function corresponds to a consistent estimate. 
	Substitution of the system dynamics~\eqref{eq:SystemDynamicsUnmodelledDynamics} and the PGNN~\eqref{eq:PGNN} into the identification criterion~\eqref{eq:IdentificationCriterionMSERegularized} for $N \to \infty$ gives the cost function
	\begin{align}
	\begin{split}
	\label{eq:Proof1Step1}
		\lim_{N \to \infty} & \frac{1}{N} \sum_{t = 0}^{N-1} \bigg( f_{\textup{phy}} \big( \theta_{\textup{phy}}^*, \phi(t) \big) - f_{\textup{phy}} \big( \theta_{\textup{phy}}, \phi(t) \big) \\
		& + g \big( \phi(t) \big) - f_{\textup{NN}} \big( \theta_{\textup{NN}}, \phi(t) \big) \bigg)^2 + \sigma_v^2 \\
		& + (\theta_{\textup{phy}}^* - \theta_{\textup{phy}}) \Lambda ( \theta_{\textup{phy}}^* - \theta_{\textup{phy}} ) \geq \sigma_v^2, 
	\end{split}
	\end{align}
	where $\sigma_v^2$ occurs from $v(t)$ which is taken outside of the MSE term, since it is zero mean white noise and uncorrelated with the regressor $\phi(t)$. 
	Since the MSE and regularization terms are non--negative, the inequality in~\eqref{eq:Proof1Step1} holds with equality only if $\theta_{\textup{phy}} = \theta_{\textup{phy}}^*$ (regularization term), and $\theta_{\textup{NN}} = \theta_{\textup{NN}}^*$ (MSE term after substitution of $\theta_{\textup{phy}} = \theta_{\textup{phy}}^*$). 
\end{proof}

\begin{remark}
	It follows from Proposition~\ref{prop:PGNNConsistencyARX} that $f_{\textup{phy}} \big( \hat{\theta}_{\textup{phy}}, \phi(t) \big) + f_{\textup{NN}} \big( \hat{\theta}_{\textup{NN}}, \phi(t) \big) = h \big( \phi(t) \big)$ for all $\phi(t) \in \mathcal{R}$. 
	Therefore, the identified PGNN perfectly replicates the system under the listed assumptions.
\end{remark}

Note that, even though Assumptions~\ref{as:ModelSet}, \ref{as:PersistenceOfExcitation}, \ref{as:GlobalOptimum} may not hold in general, the result of Proposition~\ref{prop:PGNNConsistencyARX} offers an additional reliability for the forward identification approach compared to the inverse identification.
Indeed, in the latter approach a bias is present, due to the correlation between $\phi'(t)$ and $v(t)$, and the fact that the noise $v(t)$ therefore cannot be taken out of the MSE term.

\begin{remark}
\label{re:1OE}
	Consider the system with NOE noise~\eqref{eq:SystemDynamicsOE}, see also Fig.~\ref{fig:NoiseStructures}. 
	Then, consistency is obtained by using the PGNN~\eqref{eq:PGNN} with $\phi(t) = [\hat{y} \big( \theta, \phi(t-1) \big), \hdots, \hat{y} \big( \theta, \phi(t-n_a) \big), u(t-n_k-1), \hdots, u(t-n_k-n_b)]^T$, provided that the initial conditions, i.e., $y(t)-v(t)$ and $u(t)$ for $t<0$, are known and $h$ is stable.
	The proof follows similar to the proof of Proposition~\ref{prop:PGNNConsistencyARX}. 
\end{remark}

\begin{remark}
\label{re:2IE}
	Consider the system with NIE noise~\eqref{eq:SystemDynamicsIE}, see also Fig.~\ref{fig:NoiseStructures}. 
	Then, the responses caused by the input $u(t)$ and the noise $v(t)$ cannot be distinguised.
	In this situation, a direct inverse identification is beneficial, due to its equivalence to forward identification with output noise. 
	Under the assumption that the inverse dynamics $h^{-1}$ in~\eqref{eq:OptimalFeedforwardRewritten} exists and is stable, identification according to criterion~\eqref{eq:InverseIdentificationCriterionRegularized} with inverse PGNN~\eqref{eq:PGNNInverse} and $\phi'(t) = [y(t+n_k+1), \hdots, y(t+n_k-n_a+1), \hat{u} \big( \theta, \phi'(t-1) \big), \hdots, \hat{u} \big( \theta, \phi' (t-n_b+1) \big)]^T$ yields consistent estimates. 
\end{remark}

\begin{remark}
\label{re:3CL}
	Assume that the data $Z^N$ is generated under closed--loop operation, e.g., using a linear feedback
	\begin{equation}
	\label{eq:ClosedLoopOperation}
		u(t) = C \big( q^{-1} \big) \big( r(t) - y(t) \big) + \Delta u(t),
	\end{equation}
	where $C(q^{-1})$ is the transfer function of the feedback controller, $q^{-1}$ the backwards shift operator, and $\Delta u(t)$ the excitation signal on the input used during data generation.
	The feedback controller introduces a correlation between the input and output and therefore, the noise as well. 
	However, the proof of Proposition~\ref{prop:PGNNConsistencyARX} remains valid, since $\phi(t)$ and $v(t)$ are uncorrelated due to the fact that $u(t-n_k-1)$ is the most recent input in $\phi(t)$.
\end{remark}

From Remarks~\ref{re:1OE}, \ref{re:2IE}, \ref{re:3CL} it becomes apparent that, in order to obtain consistent estimates, the regressor $\phi(t)$ must be chosen to appropriately account for where the noise enters the system. 
There are several approaches that can be used to remove the closed--loop induced bias when the regressor is not, or cannot be chosen appropriately. 
One example is the instrumental variable (IV) approach, which is for example used in~\cite{Nguyen2018a}. Therein, it was shown that a bias--correction factor was required to obtain consistent estimates, which required specific knowledge of the noise distribution, variance, and structure. 
Recently, the IV approach was also applied to NN--based identification in~\cite{Kon2022a}.


\section{FEEDFORWARD CONTROLLER DESIGN}
\label{sec:FeedforwardControllerDesign}

\subsection{Gradient--based inversion}
Since the PGNN~\eqref{eq:PGNN} is not analytically invertible in general, we employ a gradient--based inversion method to obtain the feedforward control signal. 
First, we shift~\eqref{eq:PGNN} $n_k+1$ steps forward in time and substitute $\theta = \hat{\theta}$. 
Let us define $V_R \big( u_{\textup{ff}} (t) \big)$ as the difference between the reference $r(t+n_k+1)$ and the predicted output $\hat{y} \big( \hat{\theta}, \phi_{\textup{ff}} (t+n_k+1) \big)$ for feedforward $u_{\textup{ff}}(t)$.
Then, $u_{\textup{ff}}(t)$ must satisfy
\begin{equation}
\label{eq:FeedforwardGradientBasedProblem}
	V_R \big( u_{\textup{ff}} (t) \big) := r(t+n_k+1) - \hat{y} \big( \hat{\theta}, \phi_{\textup{ff}}(t+n_k+1) \big) = 0.
\end{equation}
The gradient of $V_R \big( u_{\textup{ff}}(t) \big)$ is given as
\begin{align}
\begin{split}
\label{eq:FeedforwardGradientBasedProblemGradient}
	\frac{\partial V_R \big( u_{\textup{ff}}(t) \big)}{\partial u_{\textup{ff}}(t)} =& - \frac{\partial f_{\textup{phy}} \big( \hat{\theta}_{\textup{phy}}, \phi_{\textup{ff}}(t+n_k+1) \big)}{\partial u_{\textup{ff}}(t) } \\
	& - \frac{\partial f_{\textup{NN}} \big( \hat{\theta}_{\textup{NN}}, \phi_{\textup{ff}}(t+n_k+1)}{\partial u_{\textup{ff}}(t) }, 
\end{split}
\end{align}
where the first term is derived from the known physics--based model, and backpropagation gives
\begin{align}
\begin{split}
\label{eq:Backpropagation}
	& \frac{\partial f_{\textup{NN}} \big( \hat{\theta}_{\textup{NN}}, \phi_{\textup{ff}}(t+n_k+1) \big)}{\partial u_{\textup{ff}} (t) } \\
	&  = \hat{W}_{l+1} \beta_l \big( x_{l-1}(t) \big) \hdots \hat{W}_1 \frac{\partial \phi_{\textup{ff}}(t+n_k+1)}{\partial u_{\textup{ff}}(t) } ,
\end{split}
\end{align}
with $\beta_i \big(x_{i-1}(t) \big) := \frac{\partial \alpha_i(x)}{\partial x} \big|_{x = x_{i-1}(t)}$ with $x_{i-1}$ the output of layer $i-1$ in the NN.
Then, a gradient--based iterative search is performed to find the feedforward input $u_{\textup{ff}}(t)$ based on Algorithm~\ref{alg:FeedforwardGradientBased}.

\begin{algorithm}
\caption{Search algorithm for feedforward $u_{\textup{ff}}(t)$.}\label{alg:FeedforwardGradientBased}
\begin{algorithmic}
\State \textbf{Initialize} $u_{\textup{ff}}^{(0)}(t) = f_{\textup{phy}}^{-1} \big( \theta_{\textup{phy}}^*, \phi_{\textup{ff}}'(t) \big)$,
\For{ $i \in \{ 1, \hdots , k \}$}
	\State Compute $V_R'\big(u_{\textup{ff}}^{(i-1)}(t) \big) = \frac{\partial V_R \big(u_{\textup{ff}}(t) \big)}{\partial u_{\textup{ff}}(t)} \big|_{u_{\textup{ff}}(t) = u_{\textup{ff}}^{(i-1)}(t)}$,
	\State Update $u_{\textup{ff}}^{(i)} = u_{\textup{ff}}^{(i-1)}(t) - \frac{V_R' \big( u_{\textup{ff}}^{(i-1)}(t) \big)}{V_R \big( u_{\textup{ff}}^{(i-1)}(t) \big)}$.
\EndFor
\State \textbf{Return} $u_{\textup{ff}}(t) = \textup{arg} \min_{i \in \{0, \hdots , k \}} \big| V_R \big( u_{\textup{ff}}^{(i)}(t) \big) \big|$. 
\end{algorithmic}
\end{algorithm}

In Algorithm~\ref{alg:FeedforwardGradientBased} the search is started from the physics--based feedforward, because it is often close to the optimal feedforward for the PGNN, i.e., the output of the NN is small compared to the physics--based model, see~\cite{Bolderman2021}. 
The updates are performed using the Newton--Raphson method, since it generally converges in a limited number of iterations, but other optimization methods can also be used.
Generally, it is desired to have the number of iterations $k$ large, to ensure that the solver converges. However, for real--time implementation, the number of updates is limited by the computation time, hardware, and sampling time of the system. 
If the solver must be stopped before it converges, the last solution may not be optimal compared to previous interations. Hence, in the last step of Algorithm~\ref{alg:FeedforwardGradientBased}, we output the best feedforward signal over all iterations performed. 
Actuator limitations can be accommodated for by limiting the search within a specified domain, or by saturating the feedforward input $u_{\textup{ff}}(t)$.

\subsection{Analytical inversion}
\label{sec:AnalyticInverse}
Inversion of the PGNN~\eqref{eq:PGNN} is obstructed by the fact that the most recent input $u(t-n_k-1)$ passes through both the physical model, as well as the NN. 
This problem can be circumvented for a specific class of PGNNs, for which the NN does not use this input, i.e.,
\begin{align}
\begin{split}
\label{eq:PGNNInvertible}
	\hat{y} \big( \theta, \phi(t) \big) = f_{\textup{phy}} \big( \theta_{\textup{phy}}, \phi(t) \big) + f_{\textup{NN}} \left( \theta_{\textup{NN}}, \begin{bmatrix} \phi_y(t) \\ \phi_u(t) \end{bmatrix} \right),
\end{split}
\end{align}
where $\phi_y = [y(t-1), \hdots, y(t-n_a)]^T$ and $\phi_u (t) = [u(t-n_k-2), \hdots, u(t-n_k-n_b)]^T$, such that $\phi(t) = [\phi_y(t)^T, u(t-n_k-1), \phi_u(t)^T]^T$. 
Then, similar to the physics--based feedforward~\eqref{eq:FeedforwardPhysicsBased}, the PGNN feedforward obtained from~\eqref{eq:PGNNInvertible} is given as
\begin{align}
\begin{split}
\label{eq:PGNNInvertibleFeedforward}
	u_{\textup{ff}} (t) & = f_{\textup{phy}}^{-1} \left( \hat{\theta}_{\textup{phy}}, \phi_{\textup{ff}}'- \Delta_f (t) \right), \\
	\Delta_f(t) & = \begin{bmatrix} 1 \\ 0 \\ \vdots \\ 0 \end{bmatrix} f_{\textup{NN}} \left( \hat{\theta}_{\textup{NN}}, \begin{bmatrix} r(t+n_k) \\ \vdots \\ r(t+n_k-n_a+1) \\ u_{\textup{ff}}(t-1) \\ \vdots \\ u_{\textup{ff}}(t+n_k-n_b+1) \end{bmatrix} \right).
\end{split}
\end{align}

In the remainder of this section, we show that the PGNN~\eqref{eq:PGNNInvertible} can yield consistent estimates. 
First, we choose $\theta_{\textup{phy}} = [\zeta^T, \psi^T]^T$, where $\psi$ are the parameters that affect $u(t-n_k-1)$, and $\zeta$ are the remaining parameters. 
Then, we assume that the physical model is able to capture the effect of the most recent input. 
\begin{assumption}
\label{as:GoodEnoughPhysicsBased}
	There exists a $\psi_0$ such that the unmodelled dynamics $g \big( \phi(t) \big)$ does not depend on $u(t-n_k-1)$. 
\end{assumption}

Correspondingly, if we define $\bar{\theta}_{\textup{phy}} = [{\zeta^*}^T, \psi_0^T]^T$, the dynamics~\eqref{eq:SystemDynamics} can be rewritten into
\begin{equation}
\label{eq:SystemDynamicsUnmodelledIndependentU}
	y(t) = f_{\textup{phy}} \big( \bar{\theta}_{\textup{phy}}, \phi(t) \big) + g \left( \begin{bmatrix} \phi_y(t) \\ \phi_u(t) \end{bmatrix} \right) + v(t). 
\end{equation}
Since $\psi^*$ resulting from the MSE identification~\eqref{eq:IdentificationCriterionMSE} generally differs from $\psi_0$, it is ommitted in the regularization, i.e.,
\begin{align}
\begin{split}
\label{eq:IdentificationCriterionMSERegularizedOnlyZeta}
	\hat{\theta} = \textup{arg} \min_{\theta}&  \frac{1}{N} \sum_{t = 0}^{N-1} \left( y(t) - \hat{y} \big( \theta, \phi(t) \big) \right)^2 \\
	& + (\zeta^* - \zeta)^T \Lambda (\zeta^* - \zeta).
\end{split}
\end{align}
Finally, an additional assumption is required on the data, since we no longer penalize $\psi^*-\psi$ in the cost function.
\begin{assumption}
\label{as:PersistenceOfExcitationPhysics}
	For some $\psi^A \neq \psi^B$ with $f_{\textup{phy}} \big( [{\zeta^*}^T, {\psi^A}^T]^T, \phi(t) \big) \neq f_{\textup{phy}} \big( [{\zeta^*}^T, {\psi^A}^T]^T, \phi(t) \big)$ for some $\phi(t) \in \mathcal{R}$, we have that
	\begin{equation}
	\label{eq:PersistenceOfExcitationPhysics}
		\frac{1}{N} \sum_{t =0}^{N-1} \left( f_{\textup{phy}} \left( \begin{bmatrix} \zeta^* \\ \psi^A \end{bmatrix}, \phi(t) \right) -f_{\textup{phy}} \left( \begin{bmatrix} \zeta^* \\ \psi^B \end{bmatrix}, \phi(t) \right) \right)^2 > 0. 
	\end{equation}
\end{assumption}

\begin{proposition}
	Consider the PGNN~\eqref{eq:PGNNInvertible} that is used to identify the system~\eqref{eq:SystemDynamics} according to identification criterion~\eqref{eq:IdentificationCriterionMSERegularizedOnlyZeta}. 
	Suppose that Assumptions~\ref{as:ModelSet}, \ref{as:PersistenceOfExcitation}, \ref{as:GlobalOptimum}, \ref{as:GoodEnoughPhysicsBased}, and~\ref{as:PersistenceOfExcitationPhysics} hold.
	Then, for $N \to \infty$, the PGNN parameters are identified as $\hat{\theta} = \{ \hat{\theta}_{\textup{phy}}, \hat{\theta}_{\textup{NN}} \} = \{ \bar{\theta}_{\textup{phy}}, \theta_{\textup{NN}}^* \}$.
\end{proposition}
\begin{proof}
	The proof follows similarly to the proof of Proposition~\ref{prop:PGNNConsistencyARX}, i.e., we substitute~\eqref{eq:SystemDynamicsUnmodelledIndependentU} and~\eqref{eq:PGNNInvertible} into~\eqref{eq:IdentificationCriterionMSERegularizedOnlyZeta} and take $v(t)$ out of the MSE term to obtain the cost function
	\begin{align}
	\begin{split}
	\label{eq:Proof2Step1}
		\lim_{N \to \infty}&  \frac{1}{N} \sum_{t=0}^{N-1} \bigg( f_{\textup{phy}} \big( \bar{\theta}_{\textup{phy}} , \phi(t) \big) - f_{\textup{phy}} \big( \theta_{\textup{phy}} \big( \theta_{\textup{phy}}, \phi(t) \big) \\
		& + g \left( \begin{bmatrix} \phi_y(t) \\ \phi_u(t) \end{bmatrix} \right) - f_{\textup{NN}} \left( \theta_{\textup{NN}} , \begin{bmatrix} \phi_y(t) \\ \phi_u(t) \end{bmatrix} \right) \bigg)^2 + \sigma_v^2 \\ 
		& + (\zeta^* - \zeta)^T \Lambda ( \zeta^* - \zeta) \geq \sigma_v^2. 
	\end{split}
	\end{align}
	In~\eqref{eq:Proof2Step1} the equality holds only if $\zeta = \zeta^*$ (regularization term), and $\{ \psi, \theta_{\textup{NN}} \} = \{ \psi_0, \theta_{\textup{NN}}^* \}$ (MSE term, substitute $\zeta = \zeta^*$ and observe that, when $\psi \neq \psi_0$ the physical model mismatch in the MSE cannot be compensated for by the NN, since it does not take $u(t-n_k-1)$ as input). 
\end{proof}

\begin{figure}
\centering
\includegraphics[width=1.0\linewidth]{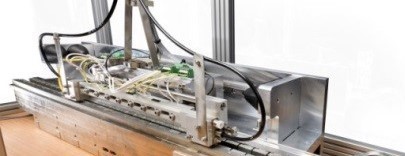}
\caption{Experimental coreless linear motor setup.}
\label{fig:CLM}
\end{figure}

\section{EXPERIMENTAL VALIDATION}
\label{sec:ExperimentalValidation}
Effectiveness of the developed PGNN feedforward controllers is validated on the problem of closed--loop position control for the real--life coreless linear motor (CLM) in Fig.~\ref{fig:CLM} that is also considered in~\cite{Bolderman2022}, see~\cite{Bolderman2021} for details on the CLM and the feedback controller. 
Data is generated by sampling the input and output at a frequency of $1$ $kHz$ for the duration of $120$ $s$ while exciting the CLM with a normally distributed white noise on the input $\Delta u(t) \sim \mathcal{N} (0, 50^2)$ $N$ in combination with a third order reference $r(t)$ that oscillates in $r(t) \in \{-0.1, 0.1 \}$ $m$ with maximum velocity $\max \big( |\dot{r}(t)| \big) = 0.05$ $\frac{m}{s}$, acceleration $\max \big( | \ddot{r}(t) | \big) = 4$ $\frac{m}{s^2}$, and jerk $\max \big( | \dddot{r}(t) | \big) = 1000$ $\frac{m}{s^3}$.

The following physics--based model is derived using Newton's second law
\begin{equation}
\label{eq:CLMDynamics}
	\delta^2 y(t) = - \frac{f_v}{m} \delta y(t) -\frac{f_c}{m} \textup{sign} \big( \delta y(t) \big) + \frac{1}{m} u(t) ,
\end{equation}
with $\delta$ a discrete--time differential operator, e.g., backward Euler $\delta = \frac{1-q^{-1}}{T_s}$, and $m$, $f_v$, and $f_c$ the mass, viscous friction coefficient, and Coulomb friction coefficient, respectively. 
We consider the following approaches to feedforward control:
\begin{enumerate}
	\item \emph{Direct inverse}, i.e., PGNN~\eqref{eq:PGNNInverse} trained according to~\eqref{eq:InverseIdentificationCriterionRegularized}. Essentially, the PGNN as proposed in~\cite{Bolderman2021}.
	\item \emph{Indirect optimization--based inverse}, i.e., PGNN~\eqref{eq:PGNN} trained according to~\eqref{eq:IdentificationCriterionMSERegularized}, with the feedforward resulting from Algorithm~\ref{alg:FeedforwardGradientBased} using $k = 5$ iterations. 
	\item \emph{Indirect analytical inverse}, i.e., PGNN~\eqref{eq:PGNNInvertible} trained according to~\eqref{eq:IdentificationCriterionMSERegularizedOnlyZeta} with feedforward~\eqref{eq:PGNNInvertibleFeedforward}. 
\end{enumerate}
All PGNNs follow the NARX structure, and have one hidden layer with $16$ $\tanh$--neurons.
Training is performed using the Levenberg--Marquardt algorithm with $\Lambda = 0.01 I$ in the identification criteria.
For each configuration, the PGNN was selected that achieved the smallest converged cost function out of $10$ trainings with random weight initialization.

\begin{figure}
	\begin{subfigure}{1\linewidth}
	\centering
	\includegraphics[width=1.0\linewidth]{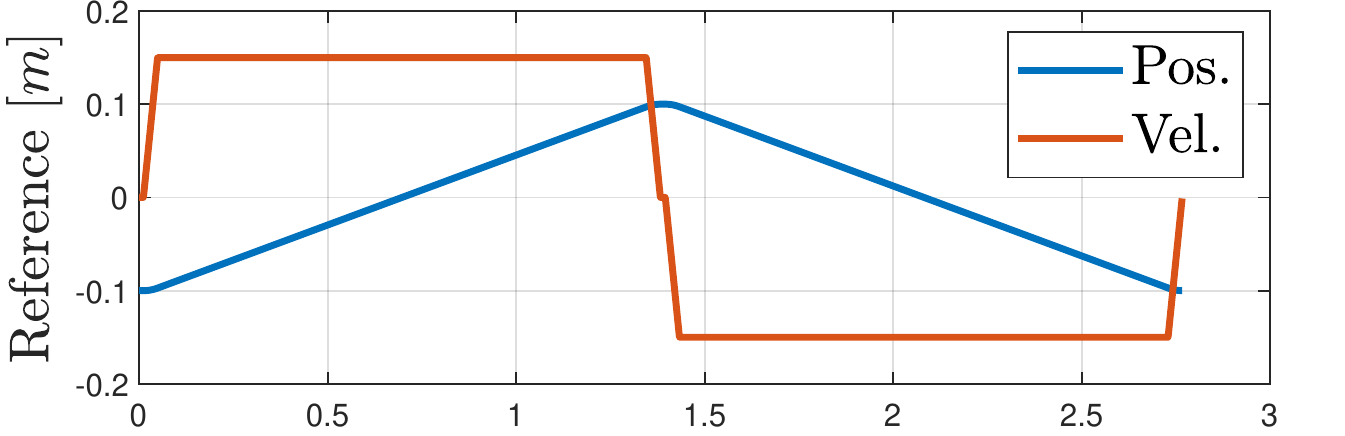}
	\end{subfigure}\hfill
	\begin{subfigure}{1\linewidth}
	\centering
	\includegraphics[width=1.0\linewidth]{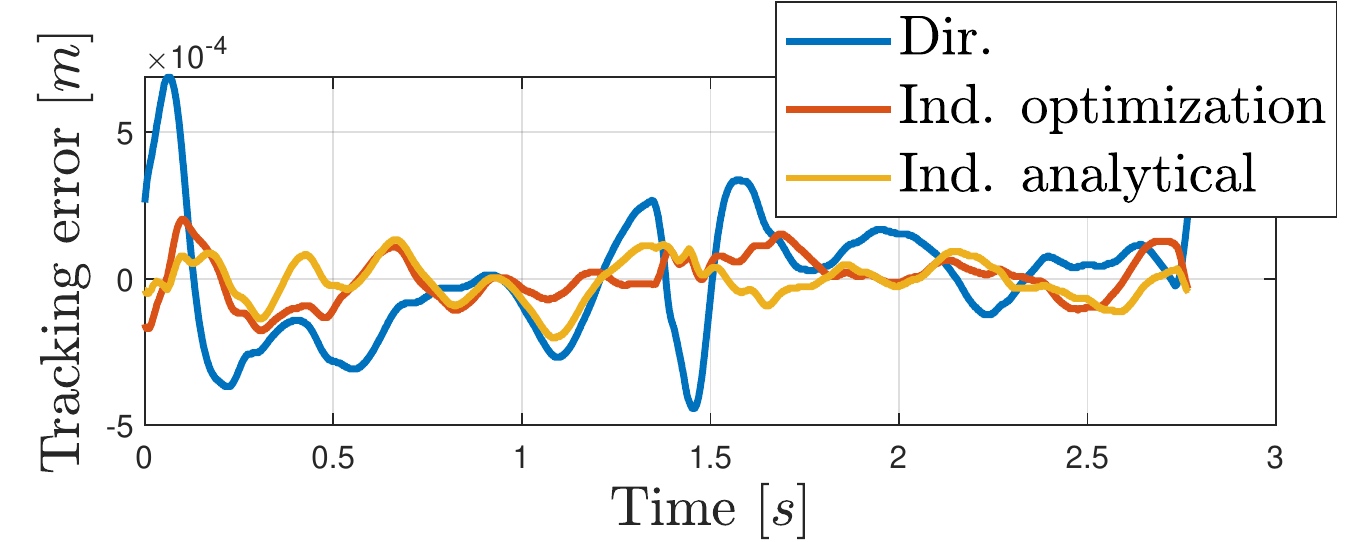}
	\end{subfigure}
	\caption{Tracking error (bottom) resulting from the reference (top) for the considered feedforward controllers. }
	\label{fig:TrackingError}
\end{figure}

Fig.~\ref{fig:TrackingError} shows the tracking error $e(t) = r(t) - y(t)$ resulting from the different PGNN feedforward controllers for the reference $r(t)$ with velocity $\max \big( \dot{r}(t) | \big) = 0.15$ $\frac{m}{s}$.
Especially during acceleration, the indirect methods proposed in this paper exhibit significantly smaller tracking errors.
Indeed, the indirect identification based PGNN feedforward controllers reduce the peak error, i.e., $\max \big( |e(t)| \big)$, with a factor of more than three.

\begin{figure}
	\centering
	\includegraphics[width=1.0\linewidth]{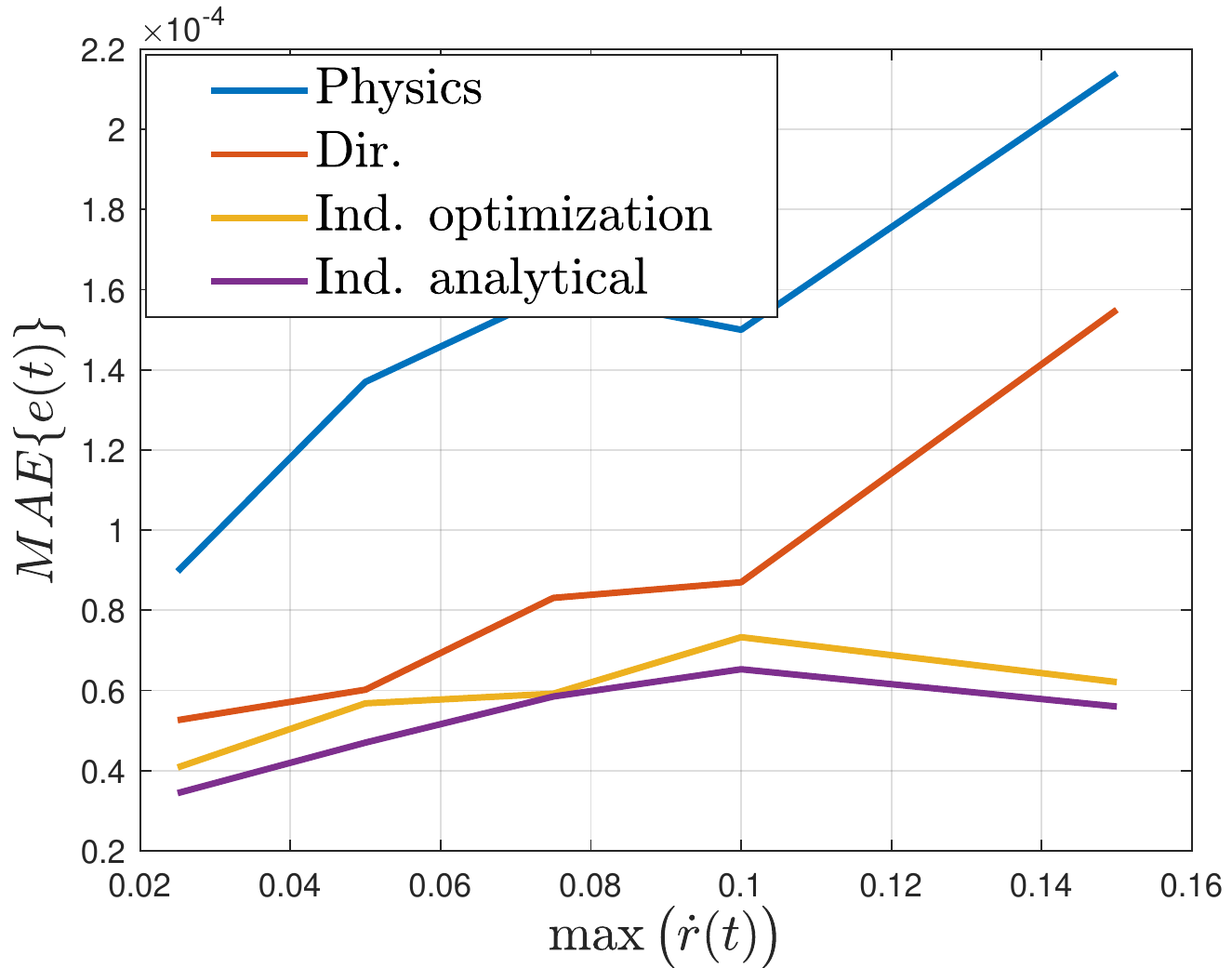}
	\caption{MAE of the tracking error evaluated on references with different maximum velocities $\max \big( \dot{r}(t) \big)$.}
	\label{fig:MAEDifferentVelocities}
\end{figure}

Fig.~\ref{fig:MAEDifferentVelocities} shows the mean--absolute error (MAE)
\begin{equation}
\label{eq:MAE}
	\textup{MAE} \big( e(t) \big) := \frac{1}{N_R} \sum_{t =0}^{N_R-1} | r(t) - y(t) |
\end{equation}
resulting from the different feedforward controllers when used on references with different maximum velocities $\max \big( | \dot{r}(t) | \big)$. 
All methods improve tracking performance with respect to the physics--based feedforward~\eqref{eq:FeedforwardPhysicsBased}.
Moreover, the indirect methods improve over the indirect method from~\cite{Bolderman2022} with around $20$--$40$\%, up until $70$\% for the velocity $0.15$ $\frac{m}{s}$.
The similarity in performance for the optimization--based and analytical feedforward indicates that the physics--based model~\eqref{eq:CLMDynamics} satisfies Assumption~\ref{as:GoodEnoughPhysicsBased} for the CLM, i.e., the CLM is linear in the input. 
The slightly better performance of the analytical inversion indicates that $k=5$ is too small for the Algorithm~\ref{alg:FeedforwardGradientBased} to converge. 
Choosing $k>5$ however, resulted in a computational load not manageable in real--time on the CLM, see Table~\ref{tab:CPUTime} which lists the computation times.
Future work will deal with computationally efficient algorithms for gradient--based inversion of PGNNs.  

\begin{table}
\centering
\caption{Mean of the computation time over the samples in the reference in Fig.~\ref{fig:TrackingError} on a $2.59$ $GHz$ Intel Core--i$7$--$9750$H using MATLAB 2019a.}
\label{tab:CPUTime}
\includegraphics[width=1.0\linewidth]{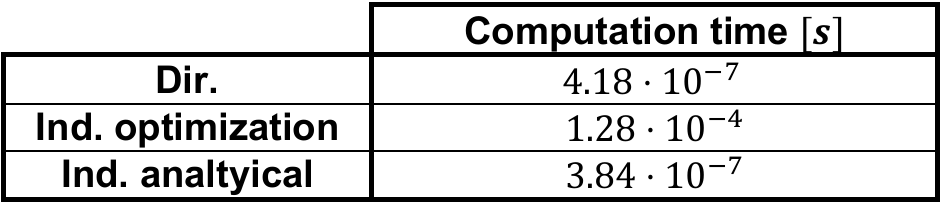}
\end{table}
 

\section{CONCLUSIONS}
\label{sec:Conclusions}
In this paper, we presented a framework for nonlinear feedforward control design in the presence of noisy data using physics--guided neurons networks. 
First, by using fundamental knowledge from the system identification field, we formulate assumptions on the PGNN model parametrization, data, and training procedure in order to obtain a consistent estimation of the forward system dynamics. 
Afterwards, two approaches are proposed for inversion of the identified PGNN describing the forward dynamics to obtain the feedforward controller. 
The developed methodology was validated on a real--life industrial linear motor, where it showed significant improvements with respect to the direct inverse approach for the same PGNN model structure.

\bibliographystyle{IEEEtran}
\bibliography{IEEEabrv,References}

\end{document}